%% file: main.tex
\documentclass[11pt]{article}

\usepackage[T1]{fontenc}
\usepackage[utf8]{inputenc}
\usepackage[english]{babel}
\usepackage{lmodern}
\usepackage{microtype}
\usepackage[a4paper,margin=2.6cm]{geometry}
\usepackage{amsmath,amssymb,amsthm,mathtools}
\usepackage{booktabs}
\usepackage{array}
\usepackage{enumitem}
\usepackage{graphicx}
\usepackage{float}
\usepackage[hidelinks]{hyperref}

\hypersetup{
  pdftitle={The Structured Totient Preimage Problem: Reconstruction, Collisions, and Cryptographic Implications},
  pdfauthor={Luis Adrian Lizama-Perez},
  pdfsubject={A structured reconstruction problem in computational number theory and its cryptographic implications},
  pdfkeywords={Euler totient, computational number theory, exact enumeration,
    multiplicative reconstruction, collisions, preimage resistance, cryptography}
}

\newtheorem{definition}{Definition}
\newtheorem{proposition}{Proposition}
\newtheorem{corollary}{Corollary}
\newtheorem{observation}{Observation}
\newtheorem{assumption}{Candidate Assumption}

\newcommand{\Pset}{\mathcal{P}}
\newcommand{\Sset}{\mathcal{S}}

\title{The Structured Totient Preimage Problem:\\
Reconstruction, Collisions, and Cryptographic Implications}
\author{Luis Adri\'an Lizama-P\'erez\\
\small Universidad Aut\'onoma Metropolitana, Unidad Lerma, Mexico\\
\small \href{mailto:l.lizama@correo.ler.uam.mx}{l.lizama@correo.ler.uam.mx}\\
\small \href{https://orcid.org/0000-0001-5109-2927}{ORCID: 0000-0001-5109-2927}}
\date{}

\begin{document}
\maketitle

\begin{abstract}
We define and study the Structured Totient Preimage (STP) problem as a restricted
reconstruction relation with a direct cryptographic motivation.  Let
$p_1,\ldots,p_k$ be distinct primes of the same bit length and reveal only
\[
  x=\prod_{i=1}^{k}(p_i-1).
\]
Given $(x,\lambda,k)$, STP asks for any set of $k$ distinct $\lambda$-bit
primes satisfying this product.  The relation is efficiently verifiable, but
its reconstruction complexity is not known.  We establish three concrete
results.  First, for factored $x$ we derive the exact number of ordered
exponent allocations and a bound showing that direct reconstruction is
polynomial for fixed $k$ when $\Omega(x)=O(\log\lambda)$; this rules out that
regime as a basis for a strong hardness claim.  Second, we give exhaustive
algorithms for reconstruction and
collision analysis.  Third, we exhaustively evaluate 28 parameter pairs, with
$2\leq k\leq5$, up to $\lambda=16$ for pairs and $4{,}588{,}935$ prime sets in
the largest census.  The data quantify non-injectivity through collision
participation, maximum multiplicity, and conditional ambiguity in bits.
These results isolate STP from general inverse-totient computation and motivate
a Structured Totient Preimage Assumption for explicitly growing parameter
families.  Under such an assumption, STP becomes a candidate
preimage-resistant relation whose implications for commitments, proofs of
knowledge of multiplicative witnesses, and authentication can be stated
precisely.  The paper establishes the computational
foundation and parameter constraints for those constructions; it does not
claim a security reduction or post-quantum hardness.
\end{abstract}



\section{Introduction}

Euler's totient function is easy to evaluate from a prime factorization.  If
$n=\prod_j q_j^{\alpha_j}$, then
\[
  \varphi(n)=\prod_j q_j^{\alpha_j-1}(q_j-1).
\]
Its inverse behavior is more complicated: some integers are not totients, the
function is not injective, and the number of inverse images varies widely
\cite{Carmichael1907,Lehmer1932,Ford1998,Ford1999}.  Consequently, recognizing
a totient, finding one inverse, and listing every inverse are different tasks.

This paper isolates a narrower relation.  Distinct primes $p_i$ are sampled
from one dyadic interval, while multiplication removes the boundaries between
the shifted factors $p_i-1$.  Reconstruction then means recovering any
partition of the multiplicative resources into blocks $d_i$ for which
$d_i+1$ are distinct primes of the prescribed size.  This restriction is not a
minor variant of asking for every solution of $\varphi(n)=x$: it fixes the
number, size, distinctness, and square-free structure of the hidden primes.

The motivation is both computational and cryptographic.  Prospective
commitment, witness-verification, and authentication mechanisms may use a
valid prime set as a secret multiplicative witness and its shifted product as
public data.  Before such protocols can be assessed, the underlying relation
must be defined independently: its witnesses, collisions, input models,
elementary algorithms, and plausible parameter regimes must be known.
Establishing that foundation is the purpose of this paper.

Four research questions organize the work:
\begin{enumerate}[leftmargin=*,itemsep=2pt,label=\textbf{RQ\arabic*.}]
  \item What is the precise search relation induced by products of shifted
        equal-bit primes?
  \item What can be proved about reconstruction when the factorization of the
        public product is supplied?
  \item How frequent and how large are collisions over exactly enumerable
        parameter ranges?
  \item Which results support, constrain, or rule out the use of the relation
        as a cryptographic hardness candidate?
\end{enumerate}

The principal contributions are:
\begin{enumerate}[leftmargin=*,itemsep=2pt]
  \item the formulation of STP as a standalone, efficiently verifiable search
        relation, with factored and unfactored input models;
  \item an exact exponent-allocation count, a reconstruction procedure, and a
        proved parameter regime in which factored-input candidate enumeration
        is polynomial;
  \item an exhaustive collision census for 28 parameter pairs, including
        explicit unique and non-unique reconstructions;
  \item an entropy-based ambiguity statistic that measures how much of a
        uniformly sampled witness remains undetermined after $x$ is revealed;
        and
  \item a candidate hardness assumption and a precise boundary between the
        results proved here and the separate cryptographic constructions they
        are intended to support.
\end{enumerate}

The contribution is therefore not a new general inverse-totient algorithm.
It is the isolation, exact finite analysis, and cryptographic delimitation of a
restricted multiplicative witness relation.  No unconditional NP-hardness,
average-case hardness, security reduction, or post-quantum claim is made.
To the best of our knowledge, prior inverse-totient work has not isolated this
joint equal-bit, distinct-prime, fixed-cardinality restriction or reported the
corresponding fiber collision census.

\section{The structured totient preimage problem}

Table~\ref{tab:notation} collects the notation used throughout the paper.  In
particular, the hidden primes $p_i$, the shifted factors $d_i=p_i-1$, and the
prime factors $q_j$ of the public value $x$ are different objects.

\begin{table}[H]
\centering
\caption{Notation used throughout the manuscript.}
\label{tab:notation}
\small
\renewcommand{\arraystretch}{1.08}
\begin{tabular}{@{}>{$}l<{$}
                >{\raggedright\arraybackslash}p{0.295\textwidth}
                >{$}l<{$}
                >{\raggedright\arraybackslash}p{0.295\textwidth}@{}}
\toprule
\multicolumn{1}{l}{Symbol} & Meaning &
\multicolumn{1}{l}{Symbol} & Meaning \\
\midrule
\lambda & Bit length of every hidden prime &
k & Number of distinct hidden primes \\
\Pset_\lambda & Set of exactly $\lambda$-bit primes &
\Sset_{\lambda,k} & Set of admissible $k$-element witnesses \\
P=\{p_1,\ldots,p_k\} & Hidden prime-set witness &
d_i=p_i-1 & Shifted factor associated with $p_i$ \\
x & Public shifted product &
q_j^{e_j} & Prime-power factors in the factorization of $x$ \\
a_{ji} & Share of exponent $e_j$ assigned to $d_i$ &
T_{\lambda,k} & Number of admissible prime sets \\
f_{\lambda,k} & Shifted-product map &
\mu_{\lambda,k}(x) & Number of witnesses in the fiber over $x$ \\
N_\lambda & Number of $\lambda$-bit primes &
V_{\lambda,k} & Number of represented public values \\
A_k(x) & Number of ordered exponent allocations &
\Omega(x) & Prime factors of $x$, counted with multiplicity \\
\mathbf{P},\mathbf{X} & Random witness and its public image &
\Delta_{\lambda,k},\ C_{\lambda,k},\ \overline{\mu}_{\lambda,k} & Ambiguity, collision participation, and mean fiber size \\
\bottomrule
\end{tabular}
\end{table}

\subsection{Relation and input models}

For an integer $\lambda\geq2$, define the set of exactly $\lambda$-bit primes
\[
  \Pset_\lambda=
  \{p\text{ prime}:2^{\lambda-1}\leq p<2^\lambda\}.
\]
For $1\leq k\leq |\Pset_\lambda|$, let
\[
  \Sset_{\lambda,k}=
  \{P\subseteq\Pset_\lambda:|P|=k\},
  \qquad
  T_{\lambda,k}=|\Sset_{\lambda,k}|=
  \binom{|\Pset_\lambda|}{k}.
\]
Define the shifted-product map
\[
  f_{\lambda,k}(P)=\prod_{p\in P}(p-1)
\]
and the multiplicity of a represented value $x$ by
\[
  \mu_{\lambda,k}(x)=
  |\{P\in\Sset_{\lambda,k}:f_{\lambda,k}(P)=x\}|.
\]

\begin{definition}[Structured totient preimage]
For given $x$, $\lambda$, and $k$, a structured preimage is a set
$P\in\Sset_{\lambda,k}$ such that $f_{\lambda,k}(P)=x$.
\end{definition}

The search task is to return any such set under the promise that one exists.
The enumeration task is to return all such sets.  These tasks are relations:
if $\mu_{\lambda,k}(x)>1$, recovering the set that originally generated $x$
is not required and in general cannot be inferred from $x$ alone.

Two input models must be distinguished.

\begin{description}[leftmargin=2.8cm,style=nextline]
  \item[Factored input.] The complete factorization
  $x=\prod_{j=1}^{r}q_j^{e_j}$ is supplied.
  \item[Unfactored input.] Only $x$, $\lambda$, and $k$ are supplied.
  A factor-first method must also factor $x$.
\end{description}

Every structured preimage $P=\{p_1,\ldots,p_k\}$ determines the square-free
integer $n=\prod_i p_i$ and satisfies $\varphi(n)=x$.  The converse does not
hold for an arbitrary inverse totient because inverse images may contain prime
powers, primes of different sizes, or a different number of prime factors.

\subsection{A different public--witness split}

The familiar RSA relation publishes $n=pq$ and keeps its factorization hidden;
$\varphi(n)$ is normally not public.  STP uses a different split.  The public
value is
\[
  x=\varphi\!\left(\prod_{p\in P}p\right)
   =\prod_{p\in P}(p-1),
\]
while the witness is a prime set $P$ satisfying explicit cardinality and
bit-length restrictions.  In the stronger factored-input model, even the prime
factorization of $x$ may be public; the remaining task is to reconstruct valid
boundaries among its prime-power resources.

The witness relation is efficiently verifiable.  Given $P$, a verifier checks
its cardinality, distinctness, bit lengths, primality, and shifted product.
This combination of easy verification and potentially difficult
reconstruction is the feature needed by the motivating cryptographic
constructions.  It is a candidate feature, not yet a hardness theorem.

\subsection{Candidate hardness statement}

Let $k=k(\lambda)$ be a polynomially bounded growing function satisfying
$2\leq k(\lambda)\leq|\Pset_\lambda|$.  On input $1^\lambda$, let
$\mathsf{GenSTP}$ sample $P$ uniformly from $\Sset_{\lambda,k(\lambda)}$ and
output $(x,\lambda,k(\lambda))$, where $x=f_{\lambda,k(\lambda)}(P)$.  The
following experiment records the exact assumption that a cryptographic
construction based on STP would need.

\begin{assumption}[Structured Totient Preimage Assumption (STPA)]
For every probabilistic polynomial-time classical algorithm $\mathcal{A}$,
\[
 \Pr\!\left[
   \begin{aligned}
   P'\in f_{\lambda,k(\lambda)}^{-1}(x)\;:\quad
   &P\mathrel{\xleftarrow{\$}}\Sset_{\lambda,k(\lambda)},\quad
     x=f_{\lambda,k(\lambda)}(P),\\[-1mm]
   &P'\leftarrow\mathcal{A}(1^\lambda,k(\lambda),x)
   \end{aligned}
 \right]
 =\operatorname{negl}(\lambda),
\]
where the probability is over the choices of $P$ and the internal randomness
of $\mathcal{A}$.
\end{assumption}

The \emph{factored STPA} gives the algorithm the complete factorization of
$x$ in canonical form as additional input.  It is the stronger candidate
assumption and is the version required by constructions that intentionally
disclose this factorization.  The results below do not prove either
assumption.  They do, however, identify a factored-input regime with only
polynomially many candidates; that regime cannot support factored STPA and
must be excluded from cryptographic parameter selection.

Table~\ref{tab:scope} separates established results from intended downstream
use.

\begin{table}[H]
\centering
\caption{Logical scope of the manuscript.}
\label{tab:scope}
\small
\begin{tabular}{>{\raggedright\arraybackslash}p{0.31\textwidth}
                >{\raggedright\arraybackslash}p{0.19\textwidth}
                >{\raggedright\arraybackslash}p{0.40\textwidth}}
\toprule
Component & Status & Role \\
\midrule
STP relation and input models & Defined here & Public statement and secret witness \\
Verification and allocation count & Proved here & Correctness and parameter screening \\
Collision and ambiguity data & Exact counts; rounded entropy & Quantifies non-injectivity \\
STPA and factored STPA & Candidate assumptions & Possible preimage-resistant foundation \\
Commitment, proof-of-knowledge, and authentication protocols & Separate constructions & Require their own security games and reductions \\
\bottomrule
\end{tabular}
\end{table}

\section{Main analytical and computational results}

\subsection{Forward collision census}

For fixed $(\lambda,k)$, the collision census is computed as follows:
\begin{enumerate}[leftmargin=*,itemsep=2pt]
  \item generate every prime in $\Pset_\lambda$ with a sieve;
  \item enumerate every unordered $k$-subset $P$;
  \item compute $x=f_{\lambda,k}(P)$; and
  \item increment the multiplicity stored for $x$.
\end{enumerate}
The result is the complete finite map $x\mapsto\mu_{\lambda,k}(x)$.  It is not
a sample.  If $N_\lambda=|\Pset_\lambda|$ and $V_{\lambda,k}$ denotes the
number of represented values, the implementation performs
$\binom{N_\lambda}{k}$ tuple evaluations and stores $V_{\lambda,k}$ integer
keys.  Its direct time cost is therefore
\[
  O\!\left(k\binom{N_\lambda}{k}\right)
\]
arithmetic multiplications, in addition to the sieve, and its dictionary
storage is $O(V_{\lambda,k})$.

\subsection{Reconstruction from a factorization}

Suppose that
\[
  x=\prod_{j=1}^{r}q_j^{e_j}
\]
is supplied in factored form.  An ordered factorization
$x=d_1\cdots d_k$ is determined by nonnegative integers $a_{ji}$ satisfying
\[
  e_j=a_{j1}+\cdots+a_{jk},
  \qquad
  d_i=\prod_{j=1}^{r}q_j^{a_{ji}}.
\]
The candidate primes are $p_i=d_i+1$.

\begin{proposition}[Exact allocation count]
\label{prop:allocation-count}
Before bit-length, primality, ordering, and distinctness tests, the number of
ordered exponent allocations is
\[
  A_k(x)=\prod_{j=1}^{r}\binom{e_j+k-1}{k-1}.
\]
\end{proposition}

\begin{proof}
For a fixed prime base $q_j$, the exponent $e_j$ has
$\binom{e_j+k-1}{k-1}$ weak allocations among $k$ positions.  Allocations for
different prime bases are independent, so their counts multiply.
\end{proof}

This is a count of candidates before rejection, not a running-time lower
bound.  Requiring $d_1<\cdots<d_k$, applying the dyadic bounds before primality
testing, and abandoning partial products that are already too large can reduce
the actual search.

\begin{proposition}[A simple allocation bound]
\label{prop:allocation-bound}
Let $\Omega(x)=\sum_{j=1}^{r}e_j$ count prime factors with multiplicity.  Then
\[
  A_k(x)\leq k^{\Omega(x)}.
\]
\end{proposition}

\begin{proof}
For one exponent $e$, every weak allocation into $k$ positions is realized by
at least one assignment of $e$ labelled copies to those positions.  There are
$k^e$ such assignments, so
$\binom{e+k-1}{k-1}\leq k^e$.  Multiplication over the prime bases gives the
claim.
\end{proof}

\begin{corollary}[A tractable factored regime]
\label{cor:tractable}
For fixed $k$, if $\Omega(x)=O(\log\lambda)$, factored-input STP search is
solvable in time polynomial in $\lambda$ by direct allocation and testing.
\end{corollary}

\begin{proof}
If $\Omega(x)\leq c\log\lambda$, then
$A_k(x)\leq k^{c\log\lambda}=\lambda^{c\log k}$, which is polynomial in
$\lambda$ for fixed $k$.  Moreover, every valid instance satisfies
$x<2^{k\lambda}$, so its candidates have $O(k\lambda)=O(\lambda)$ bits in
this regime.  Forming and testing every allocation therefore takes polynomial
time; deterministic primality testing is polynomial-time \cite{AKS2004}.
\end{proof}

This corollary is also a cryptographic constraint: a factored-input
construction cannot base its security on STP in this fixed-$k$ regime.

The reference reconstruction program implements this specialized direct procedure:
generate all weak exponent allocations, form the $d_i$, sort to remove
permutations, and retain a tuple only when every $d_i+1$ is a distinct
$\lambda$-bit prime.  For $k=2$, an equivalent procedure enumerates divisors
$d\mid x$ and tests the pair
\[
  d+1,\qquad x/d+1.
\]

\begin{observation}[Factor exposure]
If $\lambda\geq3$, $k\geq2$, and a structured preimage is returned, then
$p-1$ is a proper nontrivial divisor of $x$ for every returned prime $p$.
\end{observation}

This observation does not reduce the factorization of an arbitrary integer to
structured reconstruction.  It only records that a successful reconstruction
on a valid instance reveals a factor of that particular shifted product.

\section{Finite ambiguity measures}

Collision counts distinguish represented values, but they do not by themselves
describe how a uniformly generated input is affected.  Let the random set
$\mathbf{P}$ be uniform on $\Sset_{\lambda,k}$ and let
$\mathbf{X}=f_{\lambda,k}(\mathbf{P})$.  Conditional on $\mathbf{X}=x$, every
set in the fiber is equally likely.  Consequently,
\begin{equation}
  \Delta_{\lambda,k}
  :=H(\mathbf{P}\mid\mathbf{X})
  =\frac{1}{T_{\lambda,k}}
    \sum_x \mu_{\lambda,k}(x)\log_2\mu_{\lambda,k}(x),
  \label{eq:ambiguity}
\end{equation}
where the sum is over represented values.  This is the standard conditional
entropy identity for a deterministic map and a uniform input
\cite{Shannon1948}.  It is reported here only as a finite ambiguity measure;
it is not a security parameter.

Two additional exact summaries are used:
\begin{align*}
  C_{\lambda,k}
    &=\frac{1}{T_{\lambda,k}}
      \sum_{x:\mu_{\lambda,k}(x)\geq2}\mu_{\lambda,k}(x),
      &&\text{fraction of input sets in a non-singleton fiber},\\
  \overline{\mu}_{\lambda,k}
    &=\frac{T_{\lambda,k}}{V_{\lambda,k}},
      &&\text{mean fiber size over represented values}.
\end{align*}
These quantities answer different questions.  The maximum multiplicity records
the largest observed fiber, $C_{\lambda,k}$ measures how often a uniformly
sampled input is ambiguous, and $\Delta_{\lambda,k}$ weights each fiber by the
number of bits needed to identify an input within it.
All multiplicities and counts reported below are exact.  Decimal evaluations
of $\Delta_{\lambda,k}$ and derived ratios are computed from those exact
multiplicities and rounded to the displayed precision.

\section{Exact experiments}

\subsection{Reproducibility and range}

The program uses only the Python standard library.  It enumerates combinations
in lexicographic order, uses arbitrary-precision integer products, and stores
exact multiplicities.  Running the script rewrites both the complete CSV file
and the LaTeX table used below.  Repeated execution produced byte-identical
data and table files; entropy values are deterministic numerical evaluations
from the exact fibers and are rounded when displayed.

The census contains 28 parameter pairs:
\[
\begin{array}{c|c}
  k & \lambda\\ \hline
  2 & 4,\ldots,16\\
  3 & 6,\ldots,12\\
  4 & 6,\ldots,10\\
  5 & 7,\ldots,9.
\end{array}
\]
The largest row contains $4{,}588{,}935$ unordered prime pairs.  The complete
CSV includes the number of primes, input sets, represented values, collision
values, colliding input sets, mean fiber size, ambiguity, and maximum
multiplicity for every row.

\input{tables/collision_statistics.tex}

\subsection{Observed finite-range behavior}

Three patterns are visible in Table~\ref{tab:collision-statistics}.

First, structured collisions are not exceptional artifacts.  For example,
the fraction of colliding pairs is between $5.90\%$ and $6.66\%$ for
$12\leq\lambda\leq16$.  Over the same five rows, the conditional ambiguity is
between $0.06275$ and $0.07113$ bits.  These close finite values are empirical
observations, not an assertion of convergence.

Second, increasing $k$ can produce substantially larger fibers in the tested
ranges.  At $(\lambda,k)=(10,4)$, $33.767\%$ of the input sets belong to a
non-singleton fiber and the largest fiber contains 24 sets.  At $(9,5)$ the
corresponding values are $30.545\%$ and 17.  Comparisons across different $k$
should remain finite-range comparisons because the number of available primes
and the number of subsets change simultaneously.

Third, the rates are not monotone at small parameters.  For example, the
colliding fraction for $k=3$ drops from $12.238\%$ at $\lambda=7$ to
$6.776\%$ at $\lambda=8$, then increases again.  This behavior is one reason
not to infer an asymptotic law from a short table.

\subsection{Explicit reconstructions}

The following examples are independently checked by the reconstruction
program.

\paragraph{Unique pair.}
For $x=1440=2^5\,3^2\,5$ and $(\lambda,k)=(6,2)$, the 36 ordered exponent
allocations produce exactly one unordered structured preimage:
\[
  \{37,41\},\qquad (37-1)(41-1)=36\cdot40=1440.
\]

\paragraph{Collision.}
For $x=23400=2^3\,3^2\,5^2\,13$ and $(\lambda,k)=(8,2)$, the 72 ordered
allocations produce two structured preimages:
\[
  (131-1)(181-1)=130\cdot180=23400,
\]
\[
  (151-1)(157-1)=150\cdot156=23400.
\]
Thus a solver can reconstruct a valid pair without identifying which pair was
used to generate the target.

\paragraph{Permutation is not multiplicity.}
For $P=\{37,41,43\}$,
\[
  x=36\cdot40\cdot42=60480.
\]
The six orderings describe one set, not six preimages.  Exact enumeration finds
this set unique for $(\lambda,k)=(6,3)$.

\section{Relation to inverse-totient complexity}

The general inverse-totient problem has a deeper literature than the restricted
relation considered here.  When the factorization of the target is supplied,
Contini, Croot, and Shparlinski give a deterministic algorithm that constructs
the full inverse set in time polynomial in the number of inverse-search paths,
the divisor count, and the input length; they also obtain polynomial time for
almost all targets.  Separately, under a strong quantitative prime-pair
conjecture, they prove NP-completeness for deciding whether a finite set of
integers contains a totient \cite{Contini2006}.  Alekseyev subsequently gives a
general dynamic-programming framework for computing inverses, counts, power
sums, and extrema of multiplicative functions, with Euler's totient as a
principal example \cite{Alekseyev2016}.  The exponent-allocation procedure in
this paper should therefore be read as a transparent specialization to
square-free, equal-bit, fixed-cardinality preimages, not as a replacement for
general inverse-totient algorithms.  The isolation of the STP relation, the
restricted collision census, and its ambiguity measures are the present
paper's focus.  These general results provide important context, but they do
not directly establish hardness for a single target constrained to $k$
distinct equal-bit primes.

There is also a natural resemblance to Product Partition, which is strongly
NP-complete when the input items are indivisible \cite{Ng2010}.  In a factored
structured-preimage instance, however, prime exponents can be split among
candidate factors, and every completed factor must be one less than a prime.
The resemblance therefore does not by itself give a reduction.

The computational conclusions of this paper require neither result.  The
finite census is exact, the reconstruction algorithm is exhaustive for a
supplied factorization, and the candidate count follows directly from weak
compositions.  Whether the restricted search relation is hard in a growing
parameter regime remains open.

\section{Cryptographic consequences and boundaries}

\subsection{Candidate preimage-resistant relation}

The map $P\mapsto f_{\lambda,k}(P)$ is efficiently computable, and membership
in a fiber is efficiently verifiable.  If STPA holds for a specified growing
parameter family, this map induces a candidate preimage-resistant
\emph{relation}: the goal of an adversary is to find any valid preimage, not
necessarily the sampled one.  This distinction is essential because the
experiments prove that structured collisions occur.  Non-injectivity is
compatible with a preimage-resistant relation, but it changes how downstream
protocols must define correctness and security.

\subsection{Consequences for the motivating constructions}

The present results have direct implications for the protocols that motivated
STP.

\begin{description}[leftmargin=3.6cm,style=nextline]
  \item[Commitment mechanisms.]
  Whenever $\mu_{\lambda,k}(x)>1$, the rule ``commit with $x$ and open with any
  $P$'' is not perfectly binding: more than one valid opening exists.  This
  fact alone does not rule out computational binding, which would require a
  separate game and an analysis of whether an efficient committer can find two
  openings.  A viable construction must either enforce a uniqueness condition,
  bind auxiliary information to the sampled witness, or prove an appropriate
  computational-binding property.  The collision census therefore identifies
  a design constraint, not merely a numerical curiosity.

  \item[Proofs of knowledge.]
  For polynomially bounded $k(\lambda)$, STP supplies the explicit NP relation
  \[
    \mathcal{R}_{\mathrm{STP}}=
    \{((x,\lambda,k),P):P\in\Sset_{\lambda,k},\ f_{\lambda,k}(P)=x\}.
  \]
  It can serve as the statement--witness layer of a proof or verification
  procedure.  A concrete proof system must still specify completeness,
  knowledge soundness, zero knowledge, and efficiency.

  \item[Authentication primitives.]
  A public value $x$ can index a secret structured witness $P$, while a
  challenge--response protocol proves possession without directly releasing
  it.  Freshness, replay resistance, witness leakage, and reduction to STPA
  belong to the separate protocol analysis.
\end{description}

These applications explain why the relation is established before the
protocols are presented.  The mathematical object and its parameter limits
should not depend on one particular construction, while every later
construction can cite the same relation and state exactly which version of
STPA it assumes.

\subsection{Factored inputs and the post-quantum boundary}

The two input models lead to different security interpretations.  In the
unfactored model, ordinary factorization may contribute substantially to the
observed difficulty.  In the factored model, only the allocation-and-primality
layer remains.  Corollary~\ref{cor:tractable} proves that this residual layer is
easy in an explicit fixed-$k$ regime, so cryptographic parameters must avoid
that regime.

Quantum factoring removes the first layer \cite{Shor1997}.  Consequently, a
post-quantum claim could rest only on factored STPA in a growing parameter
family, together with a separate quantum attack analysis.  Neither requirement
is established here.
This boundary prevents the computational evidence from being misread as a
post-quantum security result.

\subsection{Limitations and next tests}

The census is exact but finite and therefore proves no asymptotic collision
law.  Direct subset enumeration scales combinatorially and cannot reach
cryptographic bit lengths.  The data establish neither average-case nor
worst-case hardness, and no complete protocol or security game is included.

The next mathematical-computational tests are now specific: study growing
$k(\lambda)$, measure exponent-allocation tree width under
$\mathsf{GenSTP}_{\lambda,k}$, compare branch-and-bound and meet-in-the-middle
reconstruction, and search for reductions or attacks against factored STPA.
The downstream protocol papers can then use the surviving parameter families
and supply their own security definitions.

\section{Conclusion}

This paper introduced STP as a standalone computational relation: given the
product of shifted, distinct, equal-bit primes, recover any compatible prime
set.  The relation is narrower than general inverse-totient computation and
separates two possible sources of difficulty: factoring the public product and
reconstructing admissible shifted-prime boundaries.

The main results are identifiable and complementary.  The exact allocation
formula characterizes the direct factored-input search space.  Its upper bound
proves a polynomial reconstruction regime and thereby excludes weak
cryptographic parameter choices.  Exhaustive enumeration over 28 parameter
pairs establishes that the map is non-injective and quantifies its fibers by
collision participation, maximum multiplicity, and conditional ambiguity.
The supplied algorithms and complete data make every finite claim
reproducible.

These results do not prove STPA, but they give it a precise object, generator,
success condition, and parameter boundary.  They also show why collisions
matter operationally: a naive commitment is not perfectly binding, while
proof-of-knowledge and authentication constructions must treat STP as a
relation rather than as an injective function.  The contribution is therefore
the computational foundation needed before those protocols are developed and
evaluated separately.  Establishing or refuting factored STPA for growing
parameters is the decisive next step.

\section*{Data and code availability}
The source package accompanying this article contains both Python programs,
the complete CSV data, the generated LaTeX table, and compilation instructions.
The programs use only the Python standard library.

\section*{Acknowledgment}
The author thanks Dr. Luis Hern\'andez Encinas for comments on an earlier
version of this work.

\section*{Conflict of interest}
The author declares no conflict of interest.

\bibliographystyle{plain}
\bibliography{references}

\end{document}

%% file: tables/collision_statistics.tex
\begin{table}[H]
\centering
\caption{Selected collision statistics. All counts are exact. A tuple is counted as colliding when its value \(x\) is shared by at least one other unordered tuple. The ambiguity column gives \(H(\mathbf{P}\mid\mathbf{X})\) in bits for a uniformly sampled tuple, rounded to five decimal places.}
\label{tab:collision-statistics}
\small
\resizebox{\textwidth}{!}{%
\begin{tabular}{rrrrrrrr}
\toprule
\(k\) & \(\lambda\) & \(\#\mathcal{P}_{\lambda}\) & Tuples & Values & Colliding (\%) & Ambiguity & Max. mult. \\
\midrule
2 & 6 & 7 & 21 & 21 & 0.000 & 0.00000 & 1 \\
2 & 8 & 23 & 253 & 250 & 2.372 & 0.02372 & 2 \\
2 & 10 & 75 & 2,775 & 2,684 & 6.450 & 0.06640 & 3 \\
2 & 12 & 255 & 32,385 & 31,266 & 6.664 & 0.07113 & 5 \\
2 & 14 & 872 & 379,756 & 367,181 & 6.363 & 0.06833 & 6 \\
2 & 16 & 3,030 & 4,588,935 & 4,433,899 & 6.435 & 0.07027 & 8 \\
3 & 6 & 7 & 35 & 35 & 0.000 & 0.00000 & 1 \\
3 & 8 & 23 & 1,771 & 1,710 & 6.776 & 0.06974 & 3 \\
3 & 10 & 75 & 67,525 & 60,555 & 18.713 & 0.22382 & 9 \\
3 & 12 & 255 & 2,731,135 & 2,423,283 & 19.577 & 0.25654 & 16 \\
4 & 8 & 23 & 8,855 & 8,272 & 12.660 & 0.13551 & 3 \\
4 & 10 & 75 & 1,215,450 & 967,849 & 33.767 & 0.48524 & 24 \\
5 & 8 & 23 & 33,649 & 30,175 & 19.344 & 0.21669 & 6 \\
5 & 9 & 43 & 962,598 & 787,667 & 30.545 & 0.42589 & 17 \\
\bottomrule
\end{tabular}
}
\end{table}